\documentclass{amsart}

\usepackage[foot]{amsaddr}

\usepackage{xcolor,amssymb,graphicx}
\usepackage{paralist}
\usepackage{thmtools,hyperref,cleveref}
\usepackage{nicefrac}
\usepackage{caption}

\newtheorem{lemma}{Lemma}

\crefname{proposition}{Proposition}{Propositions}
\crefname{theorem}{Theorem}{Theorems}
\crefname{figure}{Figure}{Figures}
\crefname{lemma}{Lemma}{Lemmas}
\crefname{proper}{Property}{Properties}
\crefname{section}{Section}{Sections}

\newcommand{\R}{{\rm I\kern-2pt R}}

\newcommand{\N}{{\rm I\kern-2ptN}}

\newenvironment{lemProof}{\par\noindent \emph{Proof.}}{\hfill$\triangleleft$\par\smallskip}

\graphicspath{{figures/}}

\begin{document}

\title{On the Edge-Vertex Ratio of Maximal Thrackles}


\author[Aichholzer]{Oswin Aichholzer}
\address[A1,A5]{Institute of Software Technology, Graz University of Technology, Austria}
\email[A1,A5]{[oaich,bvogt]@ist.tugraz.at}

\author[Kleist]{Linda Kleist}
\address[A2]{Technische Universit{\"a}t Braunschweig, Germany}
\email[A2]{kleist@ibr.cs.tu-bs.de}

\author[Klemz]{Boris Klemz}
\address[A3]{Institute of Computer Science, Freie Universit{\"a}t Berlin,Germany}
\email[A3]{klemz@inf.fu-berlin.de} 

\author[Schr\"oder]{Felix Schr\"oder}
\address[A4]{Technische Universit{\"a}t Berlin, Germany}
\email[A4]{fschroed@math.tu-berlin.de}

\author[Vogtenhuber]{Birgit Vogtenhuber}


%
\maketitle


\begin{abstract}
A drawing of a graph in the plane is a \emph{thrackle} if every pair of edges intersects exactly once, either at a common vertex or at a proper crossing. Conway's conjecture states that a thrackle has at most as many edges as vertices.
In this paper, we investigate the edge-vertex ratio of \emph{maximal thrackles}, that is, thrackles in which no edge between already existing vertices can be inserted such that the resulting drawing remains a thrackle. For maximal geometric and topological thrackles, we show that the edge-vertex ratio can be arbitrarily small.  When forbidding isolated vertices, the edge-vertex ratio of maximal geometric thrackles can be arbitrarily close to the natural lower bound of $\nicefrac{1}{2}$.  For maximal topological thrackles without isolated vertices, we present an infinite family with an edge-vertex ratio of~$\nicefrac{5}{6}$.
\end{abstract}

\section{Introduction}

A drawing of a graph in the plane is a \emph{thrackle} if every pair of edges intersects exactly once, either at a common vertex or at a proper crossing. Conway's conjecture from the 1960s states that a thrackle has at most as many edges as vertices~\cite{conway1972}.
While it is known that the conjecture holds true for geometric thrackles in which edges are drawn as straight-line segments~\cite{monotone}, it is widely open in general. In this paper, we investigate 
\emph{maximal thrackles}. A thrackle is \emph{maximal} if no edge between already existing vertices can be inserted such that the resulting drawing remains a thrackle. Our work is partially motivated by the results of Hajnal et al.~\cite{saturated2017} on saturated $k$-simple graphs. A graph is $k$-simple if every pair of edges has at most $k$ common points, either proper crossings and/or a common endpoint. A $k$-simple graph is \emph{saturated} if no further edge can be added while maintaining th $k$-simple property. In~\cite{saturated2017}, saturated simple graphs on $n$ vertices with only $7n$ edges are constructed, as well as saturated $2$-simple graphs on $n$ vertices with $14.5n$ edges.

If true, Conway's conjecture implies that in every thrackle the ratio between the number of edges and the number of vertices is at most 1. We denote the edge-vertex ratio of a thrackle $T$ by $\varepsilon(T)$.
In this paper, we investigate the other extreme, namely maximal thrackles with a low edge-vertex ratio.

In Section~\ref{sec:geometric}, we consider geometric thrackles.
We show that for this class the edge-vertex ratio can be arbitrarily small. This is done by a construction that allows to add isolated vertices while maintaining maximality.
 If we disallow isolated vertices, then a natural lower bound for the edge-vertex ratio is $\frac{1}{2}$. A similar construction can be used to get arbitrarily close to this bound.
\begin{restatable}{theorem}{geometric}\label{thm:geometric}
	For any $c>0$, there exist infinitely many
	\begin{compactenum}[a)]
		\item maximal geometric thrackles $T_a$ such that $\varepsilon(T_a)<c$, as well as 
		\item maximal geometric thrackles $T_b$ without isolated vertices such that  $\varepsilon(T_b)<\frac{1}{2} +c$.
	\end{compactenum} 
\end{restatable}

We then consider topological thrackles in Section~\ref{sec:topological-isolated}. Similarly as before, we show that the edge-vertex ratio can approach zero using isolated vertices.
\begin{restatable}{theorem}{withIso}\label{thm:smallDensity}
	For every $c>0$, there are infinitely many maximal thrackles $T'$ with $\varepsilon(T')<c$.
\end{restatable}
Note that \cref{thm:smallDensity} is not just a trivial implication of \cref{thm:geometric}, as a maximal geometric thrackle is not necessarily a maximal topological thrackle.
As our main result, in Section~\ref{sec:topological-family}, we show that there exists an infinite family of thrackles without isolated vertices which has an edge-vertex ratio of $\frac{5}{6}$.      
\begin{restatable}{theorem}{fourOverFive}\label{thm:4/5} 
	There exists an infinite family of thrackles $\mathcal{F}$ without isolated vertices, such that for all $T\in\mathcal{F}$ it holds that  $\varepsilon(T)=\frac{5}{6}$.
\end{restatable}
Our construction is based on an example presented by Kyn\v{c}l~\cite{DBLP:journals/dcg/Kyncl13} in the context of simple drawings where he showed that not every simple drawing can be extended to a simple drawing of the complete graph. The example was also used in~\cite{DBLP:journals/comgeo/KynclPRT15} for a related problem.\medskip

\textbf{Related Work.}
In one of the first works on Conway's Thrackle Conjecture, Woodall~\cite{woodall1969} characterized all thrackles under the assumption that the conjecture is true. For example, he showed that a cycle $C_n$ has a thrackle embedding with straight edges if and only if $n$ is odd. It is not hard to come up with other graphs on $n$ vertices with $n$ edges that have a  thrackle embedding, but adding an additional edge always seems to be impossible. Consequently, two lines of research emerged from Conway's conjecture. In the first, the goal is to prove the conjecture for special classes of drawings, while the second direction aims for upper bounds on the number of pairwise crossing or incident edges in any simple topological drawing with $n$ vertices.

For straight line drawings of thrackles, so called \emph{geometric} thrackles, already Erd\H{o}s provided a proof for the conjecture, actually answering a question from 1934 by Hopf and Pannwitz on distances between points. Probably the most elegant argument is due to Perles and can be found in~\cite{monotone}. 
Extending geometric drawings, a drawing is called \emph{$x$-monotone} if each curve representing an edge is intersected by every vertical line in at most one point. In the same paper, Pach and Sterling~\cite{monotone} show that the conjecture holds for $x$-monotone drawings by imposing a partial order on the edges.

A drawing of a graph is called \emph{outerplanar} if its vertices lie on a circle and its edges are represented by continuous curves contained in the interior of this circle. In~\cite{outerplanar} several properties for outerplanar thrackles are shown, with the final result that outerplanar thrackles are another class where the conjecture is true. Misereh and Nikolayevsky~\cite{mn2018} generalized this further to thrackle drawings where all vertices lie on the boundaries of $d \leq 3$ connected domains which are  in the complement of the drawing. They characterize annular thrackles ($d=2$) and pants thrackles ($d=3$) and show that in all cases Conway's conjecture holds.
Finally, Cairns, Koussas, and Nikolayevsky~\cite{spherical} prove that the conjecture holds for spherical thrackles, that is, thrackles drawn on the sphere such that the edges are arcs of great circles.

In a similar direction, several attempts show that some types of thrackles are \emph{non-extensible}.
A thrackle is called non-extensible if it cannot be a subthrackle of a counterexample to Conway's
conjecture. Wehner~\cite{wehner2013} stated the hypothesis that a potential counterexample
to Conway's conjecture would have certain graphtheoretic properties. Li, Daniels, and Rybnikov~\cite{li2006study} support this hypothesis by reducing Conway's conjecture to the
problem of proving that thrackles from a special class (which they call 1-2-3 group) are non-extensible. Actually, already Woodall~\cite{woodall1969} had shown that if the conjecture is false, then there exists a counterexample consisting of two even cycles that share a vertex.

On the negative side, we mention tangled- and generalized thrackles.
A tangled-thrackle is a thrackle where two edges can have a common point of tangency instead of a proper crossing. Besides the fact that tangled-thrackles with at least $\lfloor 7n/6 \rfloor$ edges are known~\cite{DBLP:conf/s-egc/PachRT11} -- and therefore Conway's conjecture can not be extended to tangled-thrackles -- Ruiz-Vargas, Suk, and T{\'o}th~\cite{ruiz2016disjoint} show that the number of edges for tangled-thrackles is $O(n)$. A \emph{generalized} thrackle is a drawing where any pair of edges shares an odd number of points. Lov\'asz, Pach, and Szegedy~\cite{lovasz1997} showed that a bipartite graph can be drawn as a generalized
thrackle if and only if it is planar. As planar bipartite graphs can have up to
$2n-4$ edges, this implies that generalized thrackles exist with a edge-vertex ratio close to 2. A tight upper bound of $2n-2$ edges for generalized thrackles was later provided by Cairns and Nikolayevsky~\cite{genThrackle}.

The race for an upper bound on the number $m$ of edges of a thrackle was started by the two just mentioned papers. Lov\'asz, Pach, and Szegedy~\cite{lovasz1997} provided the first linear bound of $m\leq2n-3$ and Cairns and Nikolayevsky~\cite{genThrackle} improved this to $m\leq\frac{3}{2}(n-1)$. They also consider more general drawings of thrackles on closed orientable surfaces; see also~\cite{Cairns2009}.

By exploiting certain properties of the structure of possible counterexamples, Fulek and Pach~\cite{FulekPach2011} gave an algorithm that, for any $c > 0$, decides whether the number of edges are at most $(1 + c)n$ for all thrackles. As the running time of this algorithm is exponential in $1/c$, the possible improvement by the algorithm is limited, but the authors managed to show an upper bound of $m\leq\frac{167}{117}n \approx 1.428n$.
Combining several previous results in a clever way, Goddyn and Xu~\cite{bounds2017} slightly improved this bound to $m\leq1.4n-1.4$. Among other observations they also used the fact that it was known that Conway's conjecture holds for $n \leq 11$.
This has been improved to $n \leq 12$ in the course of enumerating all path-thrackles for $n$ up to 12 in~\cite{pammer2014}. The currently best known upper bound of $m \leq 1.3984n$ is again provided by Fulek and Pach~\cite{FulekPach2019}. They also show that for \emph{quasi-}thrackles Conway's conjecture does not hold. A quasi-thrackle is a thrackle where two edges that do not share a vertex are allowed to cross an odd number of times. For this class they provide an upper bound of $m \leq \frac{3}{2}(n-1)$ and show that this bound is tight for infinitely many values of $n$.

\newpage
\section{Geometric thrackles}
 
\label{sec:geometric}

For maximal geometric thrackles, the edge-vertex ratio can be arbitrarily small. Even if we forbid isolated vertices, it may be arbitrarily close to the natural lower bound of $\frac{1}{2}$, which is implied by the handshaking lemma.

\geometric*

\begin{proof}
	Consider the thrackle~$T$ formed by the seven dark, thick edges in \cref{fig:butterfly1,} which we call the \emph{butterfly}.
		\begin{figure}[htb]
			\centering
			\includegraphics[page=22,scale=.9]{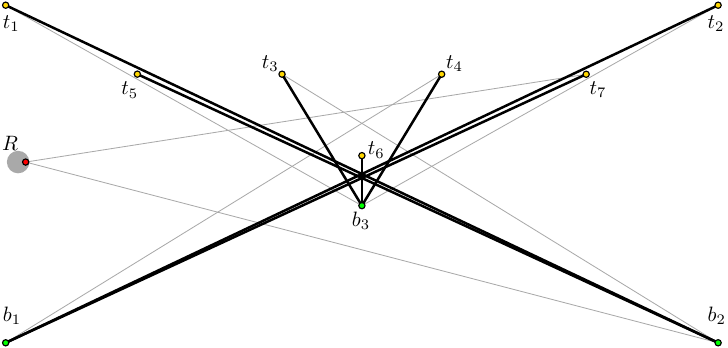}
			\caption{The butterfly~$T$ (thick, dark edges). Important segments between nonadjacent vertices are indicated in light gray. The thrackle~$T_a$ is obtained by adding multiple isolated vertices in the region~$R$.}
			\label{fig:butterfly1}
	\end{figure}
	The butterfly consists of seven edges and ten vertices.
	The lower endpoint of each edge belongs to the set~$\lbrace b_1,b_2,b_3\rbrace$ of \emph{bottom} vertices and the upper endpoint belongs to the set~$\lbrace t_1,t_2,\dots,t_7\rbrace$ of \emph{top} vertices.
	The endpoints of two independent edges~$b_1t_2$ and~$b_2t_1$ are the four corners of the bounding rectangle of the butterfly.
	These edges have a common point of intersection with the short vertical \emph{central} edge~$b_3t_6$.
	The edge~$b_1t_2$ is adjacent to another edge~$b_1t_7$, which is drawn to the right of, and very close to~$b_1t_2$ such that~$b_1t_7$ is disjoint from the segment~$b_3t_2$ and~$\mathrm{y}(t_6)<\mathrm{y}(t_7)$.
	The neighbor~$b_2t_5$ of~$b_2t_1$ is defined symmetrically.
	The edges~$b_1t_2$, $b_1t_7$, $b_2t_1$, and~$b_2t_5$ are called \emph{long} edges.
	The final two edges are adjacent to the central edge.
	Their upper endpoints~$t_3$ and~$t_4$ are placed on the line~$y=\mathrm{y}(t_7)$ to the left and right of the central edge, respectively, such that~$t_4$ (and, consequently,~$t_3$, $t_5$, and~$t_1$) are above the line~$b_1t_6$ and, symmetrically, the points $t_3,t_4,t_7$, and~$t_2$ are above~$b_2t_6$.

	The butterfly is a maximal thrackle:
	Any segment between the
	bottom vertices
	or between the
	top vertices
	is disjoint from the
	central edge
	or from one of the
	long edges.
	It remains to consider the segments that have one bottom and one top vertex as an endpoint.
	All these segments, except for~$b_1t_6$ and~$b_2t_6$, are disjoint from  the central edge or one of the long edges.
	In particular,~$b_3t_1$ is disjoint from~$b_2t_5$ by construction.
	Finally, the two remaining segments~$b_1t_6$ and~$b_2t_6$ are disjoint from~$b_3t_4$ or~$b_3t_3$, respectively.

\noindent To prove the theorem, we extend the butterfly in two different ways.
	\begin{compactenum}[a)]
		\item To obtain $T_a$ from $T$, we insert multiple isolated vertices in a small circular region~$R$ (indicated in \cref{fig:butterfly1}) that is placed to the left of~$t_6$ such that the lower tangent of~$R$ that passes through~$t_6$ is below all top vertices other than~$t_6$, and the upper tangent of~$R$ that passes through~$b_3$ is above all bottom vertices except for~$b_3$.
		
		Note that  the resulting geometric graph remains a maximal thrackle since
		each segment from a vertex in~$R$ to a vertex of~$T$ is disjoint from the central edge, except for segments to~$b_3$ or~$t_6$.
		However, any such segment is disjoint from one of the long edges.
		Moreover, any segment between distinct vertices in~$R$ is disjoint from all edges of~$T$.
		Thus, for any $n>\frac{7}{c}-10$, placing $n$ vertices in $R$ yields a thrackle $T_a$ with $\varepsilon(T_a)=\frac{7}{n+10}<c$.
		
		\item To obtain $T_b$ from $T$, we add several segments~$u_iv_i$ with~$i=1,2,\dots,m$ as indicated in \cref{fig:butterfly2}.
		\begin{figure}[htb]
			\centering
			\includegraphics[page=23,scale=.9]{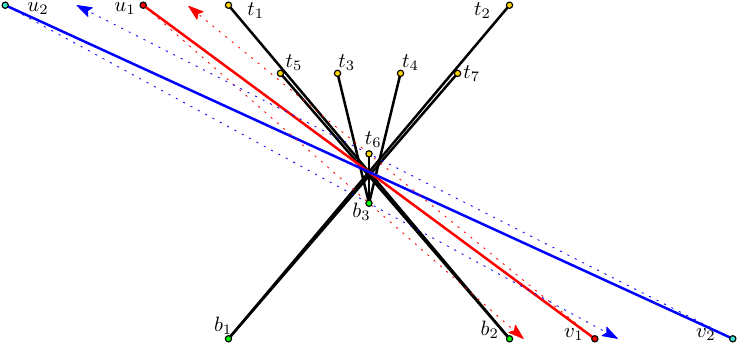}
			\caption{The thrackle~$T_b$ is obtained by adding several segments~$u_iv_i$.}
			\label{fig:butterfly2}
		\end{figure}
		All these segments pass through some common point along the central edge.
		All upper endpoints~$u_i$ are placed on the line through~$t_1$ and $t_2$, and all lower endpoints~$v_i$ are placed on the line through~$b_1$ and $b_2$.
		For each index~$i$, the slope~$\mathrm{s}(u_iv_i)$ is negative.
		Moreover, we have~$\mathrm{s}(u_iv_i)<\mathrm{s}(u_jv_j)$ for~$i<j$.
		
		Suppose that the first~$i-1$ segments have already been created for some~$i\ge 1$.
		Then we choose the slope of~$u_iv_i$ such that the vertices
		\begin{itemize}
		\item $V^+_i=\lbrace v_1,v_2,\dots,v_{i-1}\rbrace \cup \lbrace b_1,b_2\rbrace$ are below the line~$u_ib_3$; and
		\item $V^-_i=\lbrace u_1,u_2,\dots,u_{i-1}\rbrace \cup \lbrace t_1,t_2,t_3,t_4,t_5,t_7\rbrace$ are above the line~$v_it_6$.
		\end{itemize}
		These constraints can be satisfied by simply choosing~$\mathrm{s}(u_iv_i)$ large enough.
		
It is clear that by adding any number of segments in this way, we obtain a thrackle.
Moreover, this thrackle is maximal: by construction, all segments~$u_iw$ and~$v_iw$ with~$w\in V^+_i\cup V^-_i$ are disjoint from the central edge.
On the other hand, every segment with one endpoint in~$\lbrace u_i,v_i\rbrace$ and one  endpoint on the central edge is disjoint from one of the long edges.
Thus, inserting $m$ segments for some $m>\frac{1}{c}-5$ yields a maximal thrackle $T_b$ with edge-vertex ratio $\varepsilon(T_b)=\frac{7+m}{10+2m}<\frac{1}{2}+c$.
\end{compactenum} 
\vspace*{-12pt}
\end{proof}

\section{Topological thrackles of arbitrarily small edge-vertex ratio}
\label{sec:topological-isolated}

In this section, we show that the edge-vertex ratio of a maximal thrackle in the topological setting may be arbitrarily small, unless isolated vertices are forbidden.

\withIso*
\begin{proof}
	Consider the  thrackle~$T$ of a simple cycle on six vertex depicted in Figure~\ref{fig:C6:1}.
	By adding a large enough number of isolated vertices into the central triangular face~$f_0$ of~$T$, we obtain a thrackle~$T'$ with edge-vertex ratio smaller than~$c$.
	To prove the claim, it remains to show that~$T'$ is maximal.
	Towards a contradiction, assume that it is possible to insert an edge~$uv$ into~$T'$ such that the resulting graph is still a thrackle.
	Our plan is to show that~$uv$ is self-intersecting or intersects one of the edges of~$T$ twice, which yields the desired contradiction.
	To this end, we explore the drawing of~$e$, going from~$u$ to~$v$.
	
		\begin{figure}[bh]
		\centering
		\begin{minipage}{.4\textwidth}
			\centering
			\includegraphics[page=16]{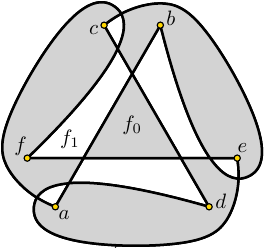}
			\caption{The thrackle~$T$.}
			\label{fig:C6:1}
		\end{minipage}\hfill
		\begin{minipage}{.4\textwidth}
			\centering
			\includegraphics[page=17]{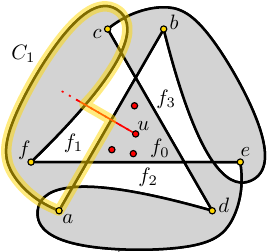}
			\caption{Case~1 in \cref{thm:smallDensity}.}
			\label{fig:C6:2}
		\end{minipage}
	\end{figure}
	
	We distinguish three cases, depending on how many of the vertices~$u,v$ are isolated vertices of~$T'$.
	
	\textbf{Case 1:} Both~$u$ and~$v$ are isolated vertices of~$T'$.
	To begin with, the edge~$uv$ has to leave~$f_0$ and, by symmetry, we may assume that it does so by intersecting~$ab$.
	The thereby entered face~$f_1$ has degree four.
	Consequently, there are three options for~$uv$ to proceed.
	
	First, assume that~$uv$ leaves~$f_1$ by intersecting the edge~$af$, as depicted in Figure~\ref{fig:C6:2}.
	By planarity, in order to reach~$v$, the edge~$uv$ has to intersect the closed curve~$C_1$ formed by parts of~$ab$ and~$af$, and the part of~$uv$ that intersects~$f_1$.
	This implies that~$uv$ intersects itself, or it intersects~$ab$ or~$af$ at least twice, which yields the desired contradiction.
	
	It follows that~$uv$ leaves~$f_1$ via~$cd$ or~$ef$.
	This implies that leaving~$f_0$ via~$f_1$ already requires crossings with two of the three segments~$ab,cd$, and~$ef$ that bound~$f_0$.
	However, traversing~$e$ in reverse, that is, going from~$v$ to~$u$, requires us to leave~$f_0$ via one of the other adjacent faces~$f_2$ and~$f_3$.
	By symmetry, this  requires two additional crossings with the segments~$ab,cd$, and~$ef$.
	Consequently, one of these segments is crossed at least twice, which again yields a contradiction.
	
	Altogether, this shows that it is not possible to add an edge between two isolated vertices of~$T'$.
	
	\textbf{Case 2:} Precisely one of~$u$ and~$v$ is isolated in~$T'$.
	Without loss of generality, we may assume that~$u$ is the isolated endpoint of~$uv$.
	As in the previous case, we may assume that~$uv$ leaves~$f_0$ via~$ab$ and enters~$f_1$.
	Given that~$uv$ has to intersect the edge~$de$ (among others), it has to leave~$f_1$ (by passing through~$af$,~$ef$, or~$cd$).
	
	The case that~$f_1$ is left via~$af$ can be excluded by planarity:
	since~$uv$ has to intersect~$de$, it has to intersect the closed curve~$C_1$, as defined in Case~1; refer to Figure~\ref{fig:C6:2}.
	Note that this holds true even if~$uv$ intersects the short piece of~$de$ located in the interior of~$C_1$, since~$de$ can only be intersected once.
	Consequently,~$uv$ intersects itself, or it intersects~$ab$ or~$af$ at least twice, which yields the desired contradiction.

	It remains to consider the cases that~$uv$ leaves~$f_1$ via~$cd$ or~$ef$, respectively.
	First, consider the former case, for an illustration refer to Figure~\ref{fig:C6:3}.
	Given that~$uv$ has already intersected~$ab$ and~$cd$, it follows that~$v\in \lbrace e,f\rbrace$.
	By planarity, it is not possible that~$v=f$, since this would imply that~$uv$ has to intersect the closed curve~$C_2$, which is composed of parts of the already intersected edges~$ab$ and~$cd$ and the edge~$af$, which is incident to~$f$.
	It follows that~$v=e$.
	At some point, the edge~$uv$ intersects the edge~$af$ in its interior and, thereby, enters the region interior to~$C_2$ that does not contain~$e$.
	However, the edges bounding~$C_2$ have now all been intersected and, hence, it is no longer possible to reach~$e$.
	
		\begin{figure}[htb]
		\centering
		\begin{minipage}{.4\textwidth}
			\centering
			\includegraphics[page=18]{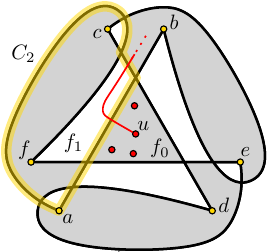}
			\caption{Case~2 in \cref{thm:smallDensity}.}
			\label{fig:C6:3}
		\end{minipage}\hfill
		\begin{minipage}{.4\textwidth}
			\centering
			\includegraphics[page=19]{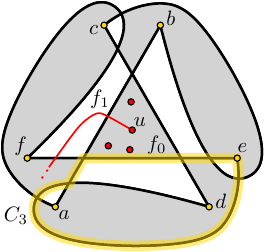}
			\caption{Case~2 in \cref{thm:smallDensity}.}
			\label{fig:C6:4}
		\end{minipage}
	\end{figure}
	
	It follows that~$uv$ does not actually leave~$f_1$ via~$cd$ and instead intersects~$ef$, for an illustration see  Figure~\ref{fig:C6:4}.
	This case can be handled very similarly to the previous one:
	since~$ab$ and~$ef$ are already intersected, it follows that~$v\in \lbrace c,d\rbrace$.
	We have~$v\neq d$, since~$d$ is enclosed by a closed curve~$C_3$ composed of parts of edges~$ab$ and~$ef$, which have already been intersected, and the edge~$de$, which is incident to~$d$.
	On the other hand, if~$v=c$, the edge~$uv$ has to intersect the edge~$de$ in its interior and, thereby, enters the region interior to~$C_3$.
	However, the edges bounding~$C_3$ have now all been intersected and, hence, it is no longer possible to reach~$c$.
	
	Overall, we obtain a contradiction, and it follows that we are actually in Case~3, to be considered next.
	
	\textbf{Case 3:} Both~$u$ and~$v$ belong to~$T$.
	Note that this implies that~$T+uv$ is a counterexample to Conways's conjecture.
		We obtain a contradiction, as it was established in the master's thesis by Pammer~\cite{pammer2014} that Conways's conjecture holds for~$n\le 12$.
	
	For the sake of self-containment, we include an independent direct proof:	
	The edge~$uv$ is a chord of~$T$.
	Consequently, the distance between~$u$ and~$v$ along~$T$ cannot be~$3$, as otherwise the thrackle
	$T+uv$ contains a thrackle of a cycle on four vertices as a subdrawing.
	It follows that the distance of~$u$ and~$v$ along~$T$ is~$2$.
	By symmetry, it suffices to consider two cases, namely~$\lbrace u,v\rbrace=\lbrace d,f\rbrace$ and~$\lbrace u,v\rbrace=\lbrace a,e\rbrace$, refer to Figures~\ref{fig:C6:5} and~\ref{fig:C6:6}, respectively.
	
	\begin{figure}[htb]
		\centering
		\begin{minipage}{.4\textwidth}
			\centering
			\includegraphics[page=20]{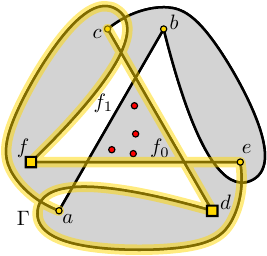}
			\caption{$\lbrace u,v\rbrace=\lbrace d,f\rbrace$.}
			\label{fig:C6:5}
		\end{minipage}\hfill
		\begin{minipage}{.4\textwidth}
			\centering
			\includegraphics[page=21]{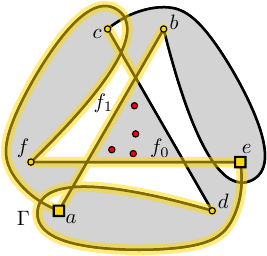}
			\caption{$\lbrace u,v\rbrace=\lbrace a,e\rbrace$.}
			\label{fig:C6:6}
		\end{minipage}
	\end{figure}
	
	Consider the subdrawing~$\Gamma$ of~$T$ formed by the edges that are incident to~$u$ and~$v$.
	The edge~$uv$ has to be drawn in the unique face~$f_{uv}$ of~$\Gamma$ that is incident to both~$u$ and~$v$.
	If~$\lbrace u,v\rbrace=\lbrace d,f\rbrace$, the part of the edge~$bc$ that passes through~$f_{uv}$ is a chord of~$f_{uv}$ that does not separate~$u$ and~$v$.
	Hence, the edge~$uv$ cannot possibly intersect~$bc$ precisely once.
	Similarly, if~$\lbrace u,v\rbrace=\lbrace a,e\rbrace$, the part of the edge~$cd$ that passes through~$f_{uv}$ is a chord of~$f_{uv}$ that does not separate~$u$ and~$v$.
	Hence, the edge~$uv$ cannot possibly intersect~$cd$ precisely once.
	Hence, in both cases we obtain a contradiction to the assumption that~$T+uv$ is a thrackle.
	
	This concludes the final case.
	Altogether, we have shown that~$T'$ is indeed maximal, which proves the claim.
\end{proof}

\section{Topological thrackles without isolated vertices}
\label{sec:topological-family}
In this section, we investigate maximal thrackles without isolated vertices, such that the edge-vertex ratio is strictly smaller than 1.  An example of such a thrackle, depicted in \cref{fig:Kyncl}, was presented by Kyn\v{c}l~\cite{DBLP:journals/dcg/Kyncl13} in the context of good drawings, i.e., drawings in which every two edges intersect at most once.

\begin{figure}[htb]
	\centering
	\includegraphics[page=12]{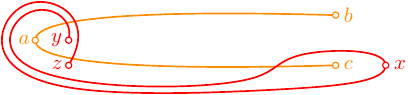}
	\caption{Kyn\v{c}l's example $K$.}
	\label{fig:Kyncl}
\end{figure}

\begin{restatable}{proposition}{KynclMax}\label{proposition:Kyncl}
	 Kyn\v{c}l's example $K$ is a maximal thrackle.
\end{restatable}
\begin{lemProof}
	We label the vertices of $K$ as depicted in \cref{fig:Kyncl}.
	$K$ is a thrackle, because the two edges incident to vertex $a$ cross the two edges incident to vertex $x$ exactly once; and the remaining pairs of edges are incident to a common vertex, namely $a$ or $x$. Now, we argue that $K$ is maximal. Suppose for a contradiction that there exists an edge $e$ such that $K\cup e$ is a thrackle.
	As Kyn\v{c}l already pointed out, an edge from $a$ to $x$ must introduce new crossings, since they don't share a common face. However, every edge of $K$ is incident to either $a$ or $x$, thus no edge is allowed to be crossed. Consequently, there is no way to draw an edge $e=ax$ such that the drawing remains a thrackle. An edge from $a$ to $y$ has to intersect exactly the edge~$xz$. In particular, an edge starting from $y$ and leaving the face via $xz$ ends in a face that is not incident to $a$; a contradiction. By symmetry of $K$, it is impossible to insert any of the other edges.
\end{lemProof}

Note that the edge-vertex ratio of Kyn\v{c}l's example is $\frac{4}{6}=\frac{2}{3}$. To date, we know of no maximal thrackle without isolated vertices that has a lower edge-vertex ratio, with the exception of $K_{1,1}$. In the following we present an infinite family of thrackles with a low edge-vertex ratio.

\fourOverFive*

\begin{proof}
We first give a high-level overview of the proof strategy.
We start our construction with a star-shaped thrackle~$T$ of the cycle $C_{2n+1}$, for some $n\geq 2$, as depicted in~\cref{fig:star} for $n=4$.
In the first step, we duplicate every vertex and edge of $T$ in  a certain way. 
This results in a thrackle drawing $T_1$ of the cycle $C_{4n+2}$. 
Then we apply another vertex/edge duplication step that consists of adding Kyn\v{c}l's example to each edge, by this obtaining a thrackle $T_2$. 
Finally, we show that if $T_2$ wasn't maximal, we could reroute an additional edge in $T_2$ to start from vertices of~$T_1$. 
Therefore, the maximality of $T_1$ implies the maximality of $T_2$. 

\begin{figure}[hb]
	\centering
	\begin{minipage}[t]{.47\textwidth}
		\centering
		\includegraphics[page=7]{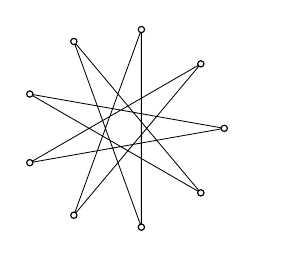}
		\captionsetup{width=\textwidth}
		\caption{Star-shaped thrackle~$T$.}
		\label{fig:star}
	\end{minipage}\hfill
	\begin{minipage}[t]{.47\textwidth}
		\centering
		\includegraphics[page=8]{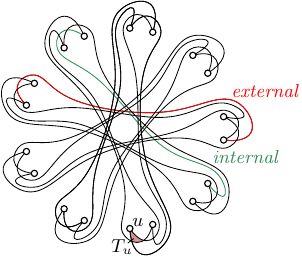}
		\captionsetup{width=\textwidth}
		\caption{The blown up star-shaped thrackle $T_1$: An external edge is highlighted in red and an internal edge in green. Moreover,  a  triangular region $T_u$ at vertex $u$ is illustrated.}
		\label{fig:starBelt}
	\end{minipage}
\end{figure}

Now, we  define $T_1$ precisely. To this end, we choose an orientation of $C_{2n+1}$ and consider three consecutive vertices $u$,$v$, and $w$ of $C_{2n+1}$.
We replace every vertex $v$ of $T$ by two vertices $v_1$ and $v_2$ very close to $v$ and on opposite sides of the cycle~$C_{2n+1}$, 
such that the \emph{first copies} $u_1$ and $v_1$ lie in the outer face of $T$ while the \emph{second copies} $u_2$ and $v_2$ lie in interior faces of $T$.
Every directed edge $uv$ of $T$ is replaced by the edges $u_1v_2$ and $u_2v_1$, which are routed in a thin \emph{tunnel} around  
$uv$ in the following way,
see~\cref{fig:Belt} for an illustration:
The edge starting at $u_1$ goes along~$uv$ 
without crossing it, surrounds~$v_1$, and then crosses the edge~$vw$ of $T$ to connect to~$v_2$.
Analogously, the edge starting at $u_2$ goes along~$uv$, surrounds $v_2$, and then crosses the edge~$vw$ of $T$ as well as $u_1v_2$ to connect to $v_1$.
The edges emanating from $v_1$ and $v_2$ are drawn analogously and hence intersect the edges~$u_1v_2$ and $u_2v_1$, respectively. 
\begin{figure}[tb]
	\centering
	\includegraphics[page=3]{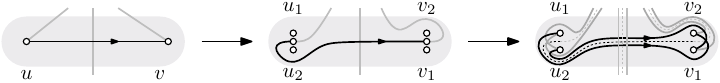} 
	\caption{Step 1: Duplicating the vertices and edges. The tunnel of $uv$ is depicted grey-shaded. For convenience we slightly bend the edges of $T$ before duplicating. In the right figure, the original cycle is indicated with dotted lines.}
	\label{fig:Belt}
\end{figure}
We will denote edges like $u_1v_2$, going from a first to a second copy of vertices in $T$ as \emph{external} edges, since these edges are the ones that are incident to the outer face of the drawing. The other edges of $T_1$, edges like $u_2v_1$, going from a second copy of a vertex in $T$ to a first copy, will be called \emph{internal} edges. 

The result $T_1$ is a drawing of the cycle $C_{4n+2}$; a drawing for $n=4$ is depicted in \cref{fig:starBelt}. 
We next show that every pair of edges $e$,$e'$ in $T_1$ intersects and hence~$T_1$ is a thrackle. 
	\begin{restatable}{lemma}{tOne}
		$T_1$ is a thrackle. 
	\end{restatable}
\begin{lemProof}
	Denote with $e_o$ and $e'_o$ the edges in $T$ from which $e$ and $e'$, respectively, originated.
	We distinguish the following cases: 
	
	\textbf{Case 1:} $e_o = e'_o$.
	This is directly handled by construction, since the two new edges replacing one edge cross.
	
	\textbf{Case 2:} $e_o$ and $e'_o$ share exactly one vertex $v$. 
	If $e$ and $e'$ share a vertex, then they intersect at that vertex and by construction, they have no further intersection. 
	If $e$ and $e'$ share no vertex, assume without loss of generality that $e_o$ is before $e'_o$ along $C_{2n+1}$. 
	Then $e$ intersects $e'$ in the tunnel of $e_o$ when it goes around the non-incident copy of~$v$.
	
	\textbf{Case 3:} $e_o$ and $e'_o$ share no vertex. 
	In this case, the tunnels of $e_o$ and $e'_o$ cross. Hence $e$ and $e'$ cross as well.
	
	Altogether, this shows that $T_1$ is a thrackle. 
\end{lemProof}
Moreover, we claim that it is maximal.
\begin{restatable}{proposition}{propTone}\label{proposition:max1}
	The  thrackle $T_1$ of $C_{4n+2}$ is maximal.
\end{restatable}

\begin{lemProof}
	To show maximality, assume for a contradiction that it is possible to insert an edge $e=st$ in $T_1$ that intersects each edge of $T_1$ exactly once.  As before, we consider $C_{4n+2}$ with an orientation.
	Let $s_o$ and $t_o$ denote the vertices in $T$ from which~$s$ and~$t$ originated, respectively. We distinguish four cases depending on the relation of $s_o$ and~$t_o$:
	
	In \textbf{Case 1},  it holds that  $s_o=t_o$, that is, $\{s, t\} = \{w_1,w_2\}$ are the two duplicates of a vertex $w:=s_o$ in~$T$. 
	Consider the three vertices $u$, $v$, and $w$ in $T$ that are consecutive in this order as well as their duplicates in $T_1$.
	The edge $e=w_1w_2$ must not cross any edge incident to $w_1$ or $w_2$ and it has to cross all other edges exactly once. 
	Hence, $e$ has to start and end in the unique face~$C$ incident to both $w_1$ and $w_2$, 
	as otherwise it is locked in a triangular face of which it must not cross any boundary edge; see~\cref{fig:maxTrackleCaseOne}.
	On the other hand, consider an edge $a$ that is not incident to any $w_i$ nor $v_i$;  $a$ crosses the edges $v_1w_2$ and $v_2w_1$ consecutively in interior points.  Note that $e$ crosses the boundary of the region~$R$  defined by parts of $a$, $v_1w_2$ and $v_2w_1$ as illustrated in ~\cref{fig:maxTrackleCaseOne},  an even number of times since it contains $C$. However, $e$ may only intersect the boundary at $a$; a contradiction.

	\begin{figure}[htb]
		\centering
		\begin{minipage}{.45\textwidth}
			\centering
			\includegraphics[page=5]{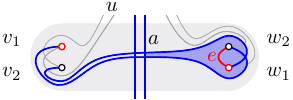}
			\captionsetup{width=.95\textwidth}
			\caption{Illustration of Case 1 of \cref{proposition:max1}.}
			\label{fig:maxTrackleCaseOne}
		\end{minipage}\hfil
		\begin{minipage}{.45\textwidth}
			\centering
			\includegraphics[page=6]{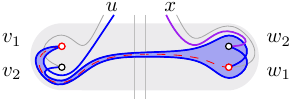}
			\captionsetup{width=.95\textwidth}
			\caption{Illustration of Case 2 of \cref{proposition:max1}.}
			\label{fig:maxTrackleCaseTwo}
		\end{minipage}
	\end{figure}
	
	In \textbf{Case 2}, the vertices $s_o$ and $t_o$ share an edge in $T$, that is, $\{s, t\} = \{v_i,w_i\}$ for some $i\in \{1,2\}$ and some directed edge $vw$ of $T$. Let $j$ fulfill $\{i,j\}=\{1,2\}$.
	Let $u$ and $x$ be the vertices preceeding~$v$ and succeeding~$w$, respectively, in $T$ and consider the duplicates of $u$, $v$, $w$, and $x$ in~$T_1$.
	Note that $w_i$ lies in the interior of a closed region $R$ bounded by parts of the edges $v_1w_2$, $v_2w_1$ and $v_iu_j$; see~\cref{fig:maxTrackleCaseTwo} for an illustration of $i=1$.
	As all those edges are incident to one of the endpoints of $e$, $e$ is completely contained in $R$.
	Moreover, $e$ must cross the edge $w_jx_i$ and hence enter the triangular face incident to $w_j$. 
	However, $w_jx_i$ is the only edge of this face that $e$ is allowed to cross, a contradiction to the existence of $e$.
	
	In Case 3 and Case 4, it holds that $s_o$ and $t_o$ have distance at least two in~$T$. Therefore, the two edges of $s$ cross the two  edges of $t$ in interior points.
	For every oriented edge of $T_1$, we define its \emph{lower part} as the section from its start vertex to the first (interior) intersection point with any edge incident to $s$ or $t$;  note that for an edge incident to $s$ the lower part is the section from its start vertex to the first (interior) intersection point with any edge incident to $t$. Likewise, we define the \emph{upper part} as the section from  the last (interior) intersection point with any edge incident to $s$ or $t$ to its end vertex.  The section of an edge between the lower and upper part  is called its \emph{middle part.}

	In \textbf{Case 3}, one of $s$ or $t$ is a first copy of some vertex $v$ in $T$. We assume that $s=v_1$ and we denote the second copy of $v$ by $v_2$. 
	Let $u$ and $w$ be the vertices preceeding and succeeding $v$, respectively, in $T$ and consider their duplicates in~$T_1$.
	Note that $s=v_1$ is contained in the region $R_1$ bounded by the upper parts of the edges $u_2v_1$, $a_1:=u_1v_2$ and one of the edges incident to $t$, say $e_t$, as illustrated in \cref{fig:maxTrackleThree}.  
	\begin{figure}[htb]
		\centering
		\includegraphics[page=7]{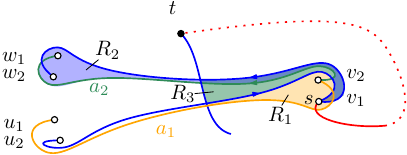}
		\caption{Illustration of Case 3 of \cref{proposition:max1}.}
		\label{fig:maxTrackleThree}
	\end{figure}
Since the edges  incident to $s=v_1$ and $t$  cannot be intersected again, the edge $e$ leaves the region $R_1$ via $a_1$. Now consider the edge $a_2:=v_2w_1$ and the region $R_2$ formed by a tiny section of $u_2v_1$, the edge $v_1w_2$ and the section of $a_2$ not contained in $R_1$. We show that it is impossible for $e$ to intersect $a_2$ exactly once.  First note that $a_2$ is covered by $R_1$ and $R_2$. 
When $e$ intersects $a_2$ within  $R_1$, it must intersect $a_1$ at least twice; a contradiction. Moreover, $R_2$ is not incident to $t$ and its boundary can only be crossed via $a_2$, which implies $e$ has to enter it directly at $s$. Then $e$ is locked in the (blue-green) region $R_3$ bounded by parts of its incident edges $v_1w_2$ and $u_2v_1$, as well as $e_t$ though, that does not contain $t$;  a contradiction. This finishes the proof of this case.
	
	In \textbf{Case 4}, both vertices $s$ and $t$ are a second copy of vertices $u$ and $v$ of~$T$, respectively. Therefore, we set $s:=u_2$ and $t:=v_2$. We may assume that the distance from $s_o$ to $t_o$ is odd in (the directed) $T$; otherwise we exchange the labels of $s$ and $t$ .
	Let~$u^1$ be the successor of $u$ in $T$, and  $u_1^1$, $u_2^1$ the two copies of  $u^1$ in $T_1$. Likewise, we denote the $i$th successor of $u$ in $T$ by $u^i$ and its $T_1$ copies $u_1^i,u_2^i$. We show that $e$ intersects the edge $a:=u_1^1u_2^2$ in its upper part: Clearly, $e$ intersects  the boundary of every region, that does neither contain $s$ nor $t$, an even number of times.  Consider the region $R$ formed by the lower and middle part of $a$ and middle and upper part of $su_1^1=u_2u_1^1$, together with a section of the incoming edge of $t=v_2$. For an illustration consider \cref{fig:maxTrackleCaseFour}.  
	\begin{figure}[htb]
		\centering
		\includegraphics[page=9]{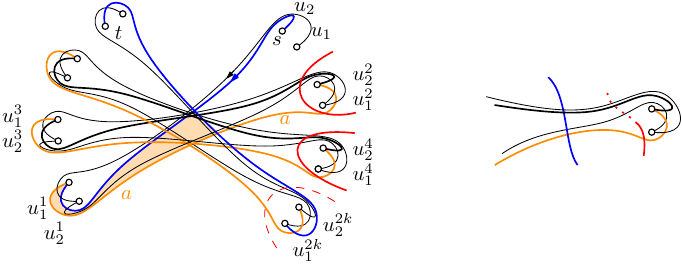}
		\caption{Illustration of Case 4 of \cref{proposition:max1}.}
		\label{fig:maxTrackleCaseFour}
	\end{figure}
	Since $R$ is not incident to $s$ nor to $t$, $e$ crosses the boundary of $R$ an even number of times. Since $e$ may cross the boundary of $R$ only via the section of~$a$; it does not intersect  $R$. Hence, $e$ intersects $a$ in its upper part.

	 The just considered edge $a$ is incident to $u_2^2$. Note that for every other edge $b$ incident  to~$u_1^2$ or~$u_2^2$, the parts  of $a$ and  $b$ incident to $u_1^2$ or~$u_2^2$, and sections of edges incident to $s$ or $t$ form regions~$R_b$ that are incident to exactly the vertices ~$u_1^2$ and~$u_2^2$. Clearly, the edge~$e$ may intersect the boundary of $R$ only at~$a$ or $b$. Consequently,  since $e$ intersects~$a$ in its upper part, $b$ is also intersected in the part incident to $u_1^2$ or $u_2^2$. Thus, for~$u^2_1$ and~$u^2_2$, all incident edges are intersected in their parts containing these vertices, upper or lower part depending on the orientation of $b$.
	
	 Recall that the distance from $s_o$ to $t_o$ is odd in $T$. Therefore, for some $k$ it holds that  $u^{2k+1}_2=t$.
	Considering~$u^{2i-1}$ and~$u^{2i}$ (for all $i=1,\dots, k$) instead of $u^1$ and $u^2$, it follows that $e$ intersects the edges incident to $u_1^{2i}$ and $u_2^{2i}$ in the parts incident to these vertices.
	This implies that the edge~$u^{2k}_1u^{2k+1}_2$ is intersected in its lower part and incident to $u^{2k+1}_2=t$; a contradiction. 
\end{lemProof}

\subsection*{Kyn\v{c}l belt construction}

For the next step, we introduce the  \emph{Kyn\v{c}l belt construction}, which
is applied to~$T_1$ in order to obtain a drawing~$T_2$.
We will show that~$T_2$ is a maximal thrackle with  edge-vertex-ratio of $\frac{5}{6}$.

The Kyn\v{c}l belt construction creates a copy $K_e$ of Kyn\v{c}l's example for each edge~$e$ of~$T_1$.
The edges of~$T_1$ are preserved and the Kyn\v{c}l copy~$K_e$ created for an edge~$e$ of~$T_1$ is drawn very close to~$e$ and interlaced with~$e$ and its incident edges, in order to ensure that the edges of~$K_e$ intersect with all edges of~$T_1$ (and~$T_2$). For an illustration consider \cref{fig:belt2-newAppen}.

\begin{figure}[htb]
	\centering
	\includegraphics[page=10]{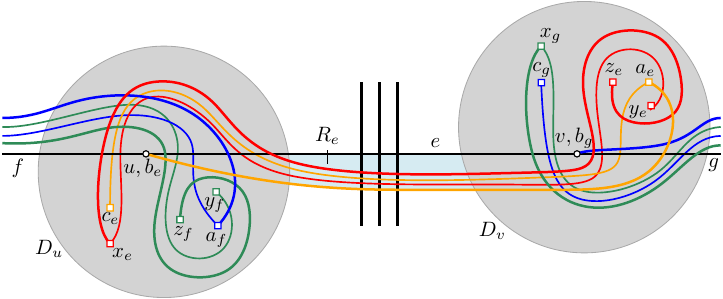}
	\caption{ Kyn\v{c}l belt construction, the original edges (thick) are preserved.}
	\label{fig:belt2-newAppen}
\end{figure}

Precisely, the construction works as follows:
for each vertex~$v$ of~$T_1$ there exists a small disk~$D_v$ containing~$v$ such that the intersection of~$D_v$ with~$T_1$ is a simple curve consisting of parts of the two edges incident to~$v$.
In particular, the disk~$D_v$ is disjoint from all edges that are not incident to $v$.
We refer to~$D_v$ as the \emph{vicinity} of~$v$.
We choose  the vertex vicinities small enough such that they are pairwise disjoint.

As in the previous step, we consider the edges of~$T_1$ to be directed. 
However, in contrast to the previous step, the construction is not symmetric and it will be useful to fix a specific orientation of the cycle; we use the orientation inherited from~$T_1$ as in \cref{fig:Belt}. It follows that external edges are oriented from a first copy of a vertex in $T$ to a second copy, whereas internal edges are oriented from a second copy to a first copy.

Consider an internal edge~$e=uv$ of~$T_1$
and let~$f$ and~$g$ denote the edges that precede and succeed~$e$ along~$T_1$, respectively.
The vertices of the Kyn\v{c}l copy~$K_e$
created for~$e$ are denoted by~$a_e,b_e,c_e$ and~$x_e,y_e,z_e$, where~$i_e$ corresponds to its pendant~$i\in \lbrace a,b,c,x,y,z\rbrace$ of Kyn\v{c}l's example illustrated in Figure~\ref{fig:Kyncl}.
The vertices~$a_e,y_e$, and~$z_e$ are placed in~$D_v$, to the left side of the directed path~$eg$, while the vertices~$c_e$ and~$x_e$ are placed in~$D_u$, to the right side of the directed path~$fe$.
Finally, the vertex~$b_e$ is identified with~$u$.

All intersections between the edges of~$K_e$ are placed inside~$D_v$ as illustrated in Figure~\ref{fig:belt2-newAppen}.
All edges of~$K_e$ cross~$g$ in~$D_v$ and then follow the edge~$e$ closely in order to reach~$D_u$.
In particular, we draw the edges close enough to~$e$ such that they are disjoint from all vertex vicinities except for~$D_v$ and~$D_u$.
Note that in this way, the edges pass through all edges of~$\mathrm{E}(T_1)\setminus \lbrace f,e,g\rbrace$.
Finally, inside~$D_u$, the edges of~$K_e$ that are non-incident to~$b_e$ cross~$e$ and then~$f$.

This construction is repeated for all internal edges and hence for every second edge of $T_1$;
recall that~$T_1$ is a cycle of even length.
For the external edges of~$T_1$, we proceed analogously, except for the fact that we use a reflected version of Kyn\v{c}l's example and we exchange the roles of the two sides of the directed paths~$eg$ and~$fe$ inside the disks~$D_u$ and~$D_v$, as illustrated in Figure~\ref{fig:belt2-newAppen}. 
This ensures three facts: 

Firstly, each edge~$e'$ of~$K_e$ crosses each edge of~$K_f$ (and~$K_g$) precisely once.
Additionally, the edges of the remaining Kyn\v{c}l copies are intersected by the part of~$e'$ that is disjoint from~$D_u$ and~$D_v$.
This shows that~$T_2$ is indeed a thrackle.

Secondly, every new vertex lies in some vicinity $D_u$ and on a specific side of the cycle $T_1$, namely inside the small triangular region $T_u$ next to vertex $u$ in $T_1$ (see \cref{fig:starBelt}). This is ensured by the fact that vertices are added to the right of the directed cycle at second copies of vertices in $T_1$ (such as $u$ in \cref{fig:belt2-newAppen}) and to the left of first copies (such as $v$ in \cref{fig:belt2-newAppen}).

Thirdly, for each edge of $T_1$, we have added four new edges and five new vertices. Consequently, the  edge-vertex-ratio is $\frac{5}{6}$. \smallskip

It remains to prove that $T_2$ is a maximal thrackle. Therefore, we assume by contradiction that there exists a new edge $s$ that can be introduced into $T_2$ such that $T_2\cup s$ is a thrackle. To arrive at a contradiction, we show a handful of properties of $s$.

To do so, we will refer to $B_e:=E(K_e)\cup \lbrace e\rbrace$ as the \emph{edge bundle of $e$}; the name is inspired by the fact that these edges run in parallel close to each other, when outside of $D_u$ or $D_v$. The \emph{region} $R_e$ of this bundle is the region of $T_2 \setminus(D_u\cup D_v)$ that is enclosed by its outer edges $e$ and $a_eb_e$, see \cref{fig:belt2-newAppen}.

\begin{lemma}\label{lem:partialbundleN}
	Let~$e=uv$ be an edge of~$T_1$.
	Suppose that the new edge~$s$ intersects at least one, but not all of the edges of
	$B_e$
	in $R_e$.
	Then we may locally reroute~$s\cap R_e$ such that each intersection of~$s$ with an edge of~$B_e$ lies inside~$D_u$ or~$D_v$.
\end{lemma}

\begin{lemProof}
	We iteratively move the intersections of~$s\cap B_e$ from~$R_e$ towards~$D_u\cup D_v$.
	
	To move one crossing, consider a maximal connected subset~$s'$ of~$s\cap R_e$ that contains a crossing of~$s$ with an edge of~$B_e$.
	At least one endpoint~$p$ of~$s'$ belongs to~$\partial D_v\cup \partial D_u$ since otherwise~$s'$ intersects all edges of~$B_e$ inside~$R_e$, which contradicts the preconditions of the Lemma.
	Assume that~$p\in \partial D_v$, the case that~$p\in \partial D_u$ can be handled analogously.
	Let~$p'$ be the first crossing of~$s$ with an edge~$e'\in B_e$ when traversing~$s'$ from~$p$ towards its other endpoint.
	
	We reroute~$s$ as follows:
	We remove the piece of~$s'$ that connects~$p$ with~$p'$.
	Instead of crossing~$e'$ at~$p'$, we follow~$e'$ closely towards~$D_v$.
	In~$D_v$, we immediately cross~$e'$ in order to reach~$p$.
	This redrawing is still a thrackle:
	clearly, an edge of~$\mathrm{E}(T_2)\setminus B_e$ intersects the redrawing of~$s'$ if and only if it intersects its unmodified version.
	Moreover, the rerouting cannot lead to self-intersection along~$s$ since the unmodified drawing of~$s$ crosses~$e'$ only once.
	
	The claim follows by iterating the above procedure.
\end{lemProof}

The rerouting steps  only affect regions of edge bundles.
Hence, once a crossing has been moved to a vertex vicinity, further rerouting steps cannot remove this crossing from its vicinity. 
Consequently, for each edge~$e$ of~$T_1$, we may assume from now on that either all of, or none of the intersections of~$s$ with~$B_e$ lie outside of vertex vicinities.

\begin{lemma}\label{lem:localmaximalN}
	Let~$e=uv$ be an edge of~$T_1$.
	If the edge~$s$ crosses all edges of $B_e$ inside $D_u \cup D_v$, then we can locally reroute~$s$ such that it does not intersect  $R_e$. 
\end{lemma}

\begin{lemProof}
	If  $s$ intersects all edges of $B_e$ inside $D_u\cup D_v$, then $s$  cannot cross any edge of $B_e$ in $R_e$.
	In particular, $s$ does not intersect the boundary edges of $R_e$, namely $e$ and $a_eb_e$. Let $s'$ be a section of $s$ inside $R_e$ between two consecutive intersections with $\partial D_u\cup\partial D_v$. We distinguish two cases to modify $s$: 
	
	First, we consider the case that $s'$ enters and leaves $R_e$ via $D_u$ (or $D_v$). Note that~$s'$ crosses no edge $e'$ inside $R_e$; otherwise $s'$ would have to cross $e'$ twice (since every edge of~$T_2$ that is not in $B_e$ and intersects $R_e$ also crosses $R_e$). We replace~$s'$ by a segment connecting its endpoints . (If this introduces self-intersections of $s$, we apply the usual modification for removing multiple edge crossings.) Note that this does not introduce new crossings, since every edge that is intersected by such a segment had to be crossed within $R_e$ if it leaves $D_u$ in between the two endpoints of $s'$. 
	
	Second, we consider the case that $s'$ enters $R_e$ via $D_u$ and leaves via $D_v$. Note that $s$ contains at most one such $s'$ and that by the first case,
	we can assume that $s$ does not leave and enter $D_u$ or $D_v$ within $R_e$.
	Let $f$ and $g$ denote the preceding and succeeding edge of $e$ in $T_1$ as in \cref{fig:belt2-newAppen}. 
	By construction of $T_2$, $s'$ crosses all edges of $T_2$ in $R_e$ except for the edges in $B_f$, $B_e$, and $B_g$.

	Now consider the part $s_v$ of $s$ that is the continuation of $s'$ after leaving $R_e$ into~$D_v$. 
	The part $s_v$ either intersects $x_gz_g$, or it crosses~$e$ or~$a_eb_e$ directly after entering~$D_v$; see again \cref{fig:belt2-newAppen}.
	In the latter case, similarly to the proof of \cref{lem:partialbundleN}, we can reroute $s$ to have all those intersections with $B_e$ just before leaving $D_u$ and then follow $B_e$ close to but outside of $R_e$ to $D_v$, which proves the claim.
	Hence it remains to consider the case when $s_v$ intersects $x_gz_g$ inside $D_v$.
	
	Our next goal is to show that $s$ has one vertex in $D_u$.
	To this end, we will first argue that~$s$ has a vertex in $D_u \cup D_r$.
	Consider the part $s_u$ of $s$ before entering $R_e$ from $D_u$, that is, the part of $s$ before $s'$.
	Let $r$ ($\neq u$) denote the vertex of $T_1$ such that $u$ and $r$ are copies of the same vertex of $T$. We will distinguish two cases regarding whether $u$ is a first or a second copy:
	
	If $u$ is a first copy, $s_u$ cannot leave the region $R_u$ bounded by parts of the two edges of $T_1$ incident to $r$, and part of any edge $h$ non-incident to $u$ and $r$, as depicted in \cref{fig:localmaximal1}~(left). 
	This follows from the fact that~$h$ and the incident edges of $r$ are crossed in $R_e$ and thus cannot be crossed to exit~$R_u$. 
	(Note that $h$ exists, since $n\geq 2$.)
	It follows that $s$ has a vertex in $D_u \cup D_r$, as claimed.

	\begin{figure}[h!]
		\begin{minipage}{.49\textwidth}
			\centering
			\includegraphics[page=1]{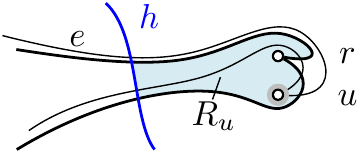}
		\end{minipage}
		\hfill
		\begin{minipage}{.49\textwidth}
			\includegraphics[page=9]{LocalMaximal.pdf}
		\end{minipage}
		\caption{Region $s$ is trapped in after leaving $R_e$ on the side of $D_u$}
		\label{fig:localmaximal1}
	\end{figure}
	
	It remains to consider the case that $u$ is a second copy.
	In this case, the edge $e$ is internal.
	Consider the region $R'_u$ bounded by parts of $f$, parts of the edges incident to $r$, and part of any edge $h$ non-incident to $u$ and $r$; see \cref{fig:localmaximal1}~(right).
	Then $s_u$ can only leave $R'_u$ by intersecting the external edge $f$ and in order to do so, it  has to intersect the whole bundle $B_f$.
	As a consequence, $s_u$ enters the outer face and, consequently, it has to cross another full edge bundle~$B_j$ to reach a vertex, see \cref{fig:starBelt}.
	All edge bundles other than~$B_f,B_e$, and~$B_g$ are crossed fully in~$R_e$.
	By assumption, the edge~$x_gz_g$ of~$B_g$ is already crossed in~$D_v$.
	It follows that~$B_j=B_e$.
	However, since~$e$ is internal, the bundle $B_e$ is not incident to the outer face, which implies that $s_u$ is trapped in the outer face.
	Therefore, $s_u$ it cannot leave $R'_u$.
	So again, it follows that $s$ has a vertex in $D_u \cup D_r$.

	Hence, independent of whether $u$ is a first or a second copy, $s$ has a vertex in $D_u \cup D_r$.
	As in both cases, every edge incident to a vertex in $D_r$ is already crossed, 
	  one vertex of $s$ lies in $D_u$. 
	
	Next, we show that $u\notin s$. Suppose by contradiction that $u\in s$.
	Then $s$ is adjacent to $a_eb_e$ and $e$. Moreover, in order to enter $R_e$, $s$ starts to the right of~$u$ between $e$ and $a_eb_e$ and hence intersects all edges of $B_f$ before entering $R_e$; see the part of the red-shaded region in \cref{fig:localmaximal2} below $e$ or \cref{fig:vicinRO}~(right).
	Now consider the continuation of $s$ after entering $D_v$. Note that $s$ cannot end at $v$ since $s$ intersects $e$ already. 
	Hence, $s$ must leave the region in $D_v$ that is bounded by $a_eb_e$, $e$, and $a_gb_g$ in order to reach another vertex; see
	 \cref{fig:localmaximal3}.
	Since $s$ intersects $a_eb_e$ and $e$ already, it must leave the region via $a_gb_g$. By doing so, $s$ crosses all edges in $B_g$.
	Consequently, all edges of $T_2$ not in $B_e$ are intersected and the other vertex of $s$ must belong to $B_e$. However, this contradicts to the maximality of Kyn\v{c}l's example, namely \cref{proposition:Kyncl}. 
	Therefore, we can conclude that $s$ has an endpoint in $D_u$ that is different from~$u$.
	
	\begin{figure}[htb]
		\centering
		\begin{minipage}[t]{.48\textwidth}
			\centering
			\includegraphics[page=6]{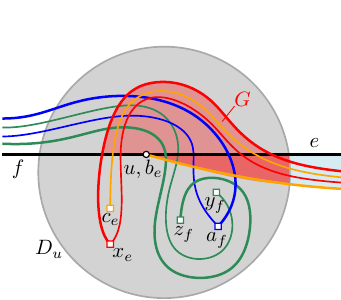}
			\captionsetup{width=.95\textwidth}
			\caption{Regions to enter $R_e$; rerouting.}
			\label{fig:localmaximal2}
		\end{minipage}
		\hfil
		\begin{minipage}[t]{.48\textwidth}
			\centering
			s		\includegraphics[page=5]{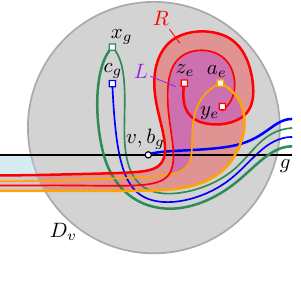}
			\captionsetup{width=.95\textwidth}
			\caption{Region to leave $R_e$.}
			\label{fig:localmaximal3}
		\end{minipage}
	\end{figure}
	
	Since $s$ is not incident to $u$, the edge $s$ enters $R_e$ by intersecting exactly one of the edges $e$ or $a_eb_e$ in~$D_u$.
	We consider the region $G$ bounded by parts of $f$, $a_eb_e$, $x_ez_e$, $e$, and the boundary of $D_u$; see \cref{fig:localmaximal2}.
	If $s$ does not enter $G$ (for the last time) via~$f$, then we reroute $s$ as follows:
	If $s$ enters $G$ for the last time via $x_ez_e\cup e$, we follow $x_ez_e\cup e$ closely outside $G$ towards $R_e$, along $R_e$ and cross the appropriate set of edges of $B_e$ shortly after entering $D_v$.
	Otherwise, if $s$ enters $G$ for the last time via $a_eb_e$, we follow $a_eb_e$ closely outside $G$ towards $R_e$, along $R_e$ and cross the appropriate set of edges of $B_e$ shortly after entering $D_v$.
	Note that this leaves the set of crossed edges (in particular, the crossed edges of~$B_f$) unchanged.

It remains to consider the case that $s$ enters $G$ (for the last time) via  $f$. This implies that $s$ intersects $e$ in $D_u$, but not $a_eb_e$.
Moreover,  $s$ intersects all edges of~$B_f$ in $G$; consequently, $s$ starts at $x_e$ or $c_e$.
Therefore, $s$ cannot end at $a_e,y_e,z_e$ by the maximality of Kyn\v{c}l's example shown in \cref{proposition:Kyncl}.
Note also that in~$R_e$, $s'$ cannot lie between $e$ and $x_ez_e$, as otherwise $e$ and $x_ez_e$, as well as all edges of~$B_g$ and $B_f$ are crossed, implying that $s$ has to end in $B_e$.
This is a contradiction to the maximality of Kyn\v{c}l's example (\cref{proposition:Kyncl}).
Consequently, $s$ enters $D_v$ in the region $R$ in $D_v$ bounded by $a_eb_e$ and $x_ez_e$, as illustrated in \cref{fig:localmaximal3}.

Moreover, the region~$L$ formed by the faces of $K_e$ incident to $a_e,y_e,$ and $z_e$ (see  \cref{fig:localmaximal3}) is not intersected by $s$:
Since $s$ must cross one of $a_eb_e$ and $x_ez_e$ when leaving $R$, at most one of them can be crossed in $L$.
Moreover, since $s$ is incident to $c_e$ or $x_e$, at most one of $x_ey_e$ and $a_ec_e$ can be crossed in $L$.
If both faces of $L$ are intersected by $s$, $s$ crosses three of the four edges which yields a contradiction.
If $s$ intersects only one face of $L$, then this is either both of  $a_eb_e$ and $x_ez_e$ or both of $x_ey_e$ and $a_ec_e$; a contradiction.
Consequently, $s$ does not intersect $L$.
 
This last fact implies that~$s$ crosses precisely one of $a_eb_e$ or $x_ez_e$ in $D_v$, namely in order to leave~$R$ (note that it is not possible for~$s$ to re-enter~$R$ later on).
Since $s$ does not intersect $a_eb_e$ in $D_u$, it follows that~$s$ leaves~$R$ via $a_eb_e$ in $D_v$ and, consequently, it intersects $x_ez_e$ in~$D_u$.

\textbf{Case 1:} If $s$ starts at $c_e$, it cannot intersect $x_ez_e$ in $G$ as otherwise it travels within the top part of $R_e$.
Thus, $s$ starts in $c_e$, intersects $x_ez_e$ and then enters~$G$ (for the last time) via~$f$.
We claim that it is impossible for $s$ to intersect $x_ey_e$.
To this end, we consider two cases regarding the order in which $x_ey_e$ and $x_ez_e$ are intersected:
If $s$ intersects $x_ey_e$ after $x_ez_e$ then  it is trapped in the region formed by the parts of $s$, $x_ez_e$, $x_ey_e$, $a_ec_e$ or it intersects~$L$; see top of \cref{fig:localmaximal4}. In both cases, we obtain a contradiction.

 \begin{figure}[htb]
	\centering
	\includegraphics[page=7]{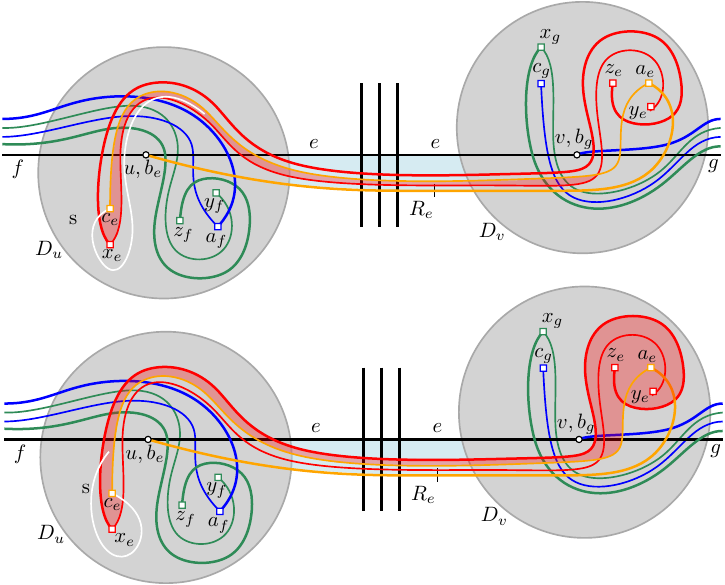}
	\caption{The two  possibilities for $s$ to start at $c_e$.}
	\label{fig:localmaximal4}
\end{figure}
It remains to consider the case that $s$ intersects $x_ey_e$ before $x_ez_e$. In this case, $s$ is trapped in the region formed by the parts of $s$, $x_ey_e$, $x_ez_e$, $a_ec_e$; see bottom of \cref{fig:localmaximal4}. Again, we obtain a contradiction and so~$s$ cannot be incident to~$c_e$.

\textbf{Case 2:} If $s$ starts at $x_e$, then it is impossible for $s'$ to be between $x_ey_e$ and $x_ez_e$ in $R_e$, since otherwise it is forced to intersect~$L$.
It follows that~$s$ intersects~$a_ec_e$ in~$D_v$.
However, this is only possible if~$s$ intersects~$L$, which again gives a contradiction.

 Altogether, we obtain that $s$ does not enter $G$ via~$f$, which  finishes the proof of \cref{lem:localmaximalN}.
\end{lemProof}

\cref{lem:localmaximalN,lem:partialbundleN} imply the following property:
\begin{restatable}{proper}{properOne}\label{proper:One}
	For every vertex $u$ and edge $e=uv$ of $T_1$ it holds that a new edge~$s$ does not enter $D_u$ within a  bundle, i.e., $s\cap R_e\cap \partial D_u=\emptyset$.
\end{restatable}

\begin{lemma}
	Let~$f=wu$ be an edge of~$T_1$. 
	If $s$ intersects any edge of $K_f$
	within $D_u$, then it intersects both the (thick) edges $a_fb_f$ and $x_fz_f$ of~$K_f$, as well as at least one of the (thin) edges~$a_fc_f$ and~$x_fy_f$ of $K_f$ within $D_u$, for an illustration refer to \cref{fig:vicinBG:2}.
	\label{lem:bluegreenN}
\end{lemma}

\begin{lemProof}
	We refer to $a_fb_f$ and $x_fz_f$ as the \emph{outer} edges, and to~$a_fc_f$ and~$x_fy_f$ as the \emph{inner} edges.
	Note that within~$D_u$, the outer edges bound a region~$R$ enclosing (parts of) the inner edges; see \cref{fig:vicinBG:2}~(left).
	Since~$s$ intersects some edge of~$K_f$ in~$D_u$, it has to pass through~$\partial R$.
	By \cref{proper:One}, the edge~$s$ does not pass through~$\partial R\cap \partial D_u$.
	
	\begin{figure}[htb]
		\centering
		\includegraphics[scale=1,page=16]{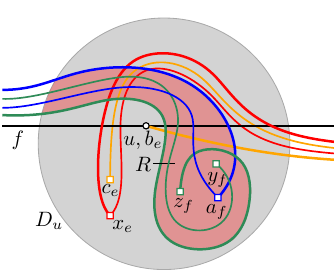}\hfill
		\includegraphics[scale=1,page=17]{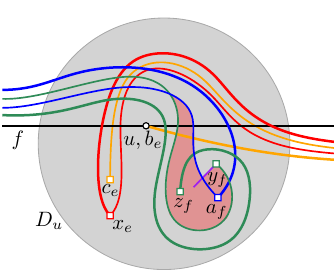}
		\caption{Illustration of \cref{lem:bluegreenN}.}
		\label{fig:vicinBG:2}
	\end{figure}
	
	We distinguish several cases regarding the endpoints of~$s$.
	
	\textbf{Case 1:}
	$s$ does not have an endpoint inside $R$.
	Then,~$s$ has to intersect the parts of \emph{both} outer edges that bound~$R$.
	The edge~$a_fc_f$ partitions the region~$R$ into two parts,  where the intersection of~$\partial R$ with one of the parts consists exclusively of a part of~$a_fb_f$, while the intersection of~$\partial R$ with the other part consists exclusively of a part of the outer edge~$z_fx_f$.
	Consequently, the edge~$s$ has to pass through~$a_fc_f$, which proves the claim.
	
	\textbf{Case 2:}
	$s$ is incident to~$a_f$.
	Hence, at this endpoint, it already intersects the outer edge~$a_fb_f$ and the inner edge~$a_fc_f$.
	Moreover, in order to leave~$R$, it has to cross the remaining outer edge, which proves the claim.
	
	\textbf{Case 3:}
	$s$ is incident to~$z_f$.
	Hence, it is already intersecting the outer edge~$z_fx_f$ and in order to leave~$R$, it also has to intersect the other outer edge.
	This, in turn, is only possible by intersecting the inner edge~$a_fc_f$, which proves the claim.
	
	\textbf{Case 4:}
	$s$ is incident to~$y_f$.
	Note that~$y_f$ is located in a triangular face formed three edges of~$K_f$.
	The edge~$s$ cannot leave this triangular region by passing through the inner edge~$a_fc_f$, since otherwise  it is trapped in the region bounded by itself and the two inner edges; see~\cref{fig:vicinBG:2}~(right).
	Therefore, the edge~$s$ has to leave the triangular region via one of the outer edges.
	Moreover, in order to leave~$R$, it has to pass through the remaining outer edge, which proves the claim.
\end{lemProof}

\begin{lemma}\label{lem:redorangeN}
	Let~$e=uv$ be an edge of~$T_1$. 
	If $s$ intersects exactly one of the inner edges $a_ec_e$ and $x_ey_e$ of $K_e$ 
	inside $D_u$, then it also intersects the outer edge $x_ez_e$ inside $D_u$. For an illustration, consider \cref{fig:vicinRO}~(left).
\end{lemma}
\begin{lemProof}
	Consider the region $R'$ defined by the sections of $x_ey_e$ and $x_ez_e$ within $D_u$ and parts of the boundary of $D_u$. Note that $R'$ encloses the edge $c_ea_e$.
	To intersect $a_ec_e$ or $x_ey_e$  within $D_u$, $s$ must intersect~$R'$.
	Further, 
	by \cref{proper:One}, $s$ does not enter or leave $R'$ via $\partial D_u$. 
	Thus, if $e$ is not incident to $x_e$ nor to~$c_e$, then it intersects $\partial R'$ twice via both  edges $x_ey_e$ and $x_ez_e$. 
	Likewise, if $e$ is incident to the vertex $x_e$ then $e$ intersects both edges $x_ez_e$ and $x_ey_e$. 
	Otherwise, $e$ is incident to the vertex $c_e$ and must leave $R'$ via $x_ez_e$ to avoid $x_ey_e$.
\end{lemProof}

\begin{figure}[htb]
	\includegraphics[scale=1,page=12]{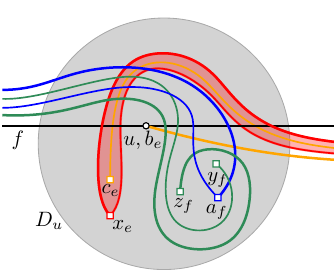}
	\includegraphics[scale=1,page=13]{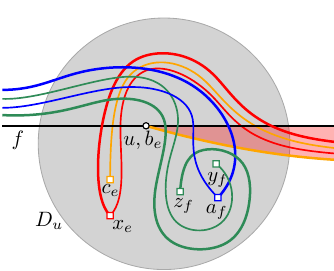}
	\caption{Illustration of \cref{lem:redorangeN} and \cref{lem:origbundleN}.}
	\label{fig:vicinRO}
\end{figure}

\begin{lemma}
	Let~$e=uv$ be an edge of~$T_1$. 
	Then $s$ intersects either both $e$ and $a_eb_e$ or none of them inside $D_u$. For an illustration, consider \cref{fig:vicinRO}~(right).
	\label{lem:origbundleN}
\end{lemma}
\begin{lemProof}
	Note that the region $R$ defined by the sections of  $e$ and $a_eb_e$ inside $D_u$ and the boundary of $D_u$ is only incident to $u=b_e$.
	The claim clearly holds if~$u\in s$.
	So suppose that~$u\notin s$.	
	Then, the claim follows from \cref{proper:One}.
\end{lemProof}

\begin{lemma}	\label{lem:completebundleN}
	Let~$e=uv$ be an edge of~$T_1$. 
	If $s$ intersects some edge of $B_e$ in~$D_u$ or~$D_v$, it intersects all edges of~$B_e$ inside this vicinity.
\end{lemma}

\begin{lemProof}
	We distinguish three cases regarding the number~$k$ of edges of~$K_e$ that are intersecting~$s$ in~$D_v$.
	
	\textbf{Case 1:} $k=0$.
	If at least one edge of~$K_e$  is intersecting~$s$ in~$R_e$, the claim follows by the assumption that either all edges of~$B_e$ are intersecting~$s$ in~$R_e$ or none (cf. \cref{lem:partialbundleN}).
	Otherwise, if no edge of~$K_e$ is intersecting~$s$ in~$R_e$, the claim follows by \cref{lem:origbundleN}. 
	
	\textbf{Case 2:} $k=4$.
	The claim follows by \cref{lem:origbundleN} and the assumption that either all edges of~$B_e$ are intersecting~$s$ in~$R_e$ or none.
	
	\textbf{Case 3:} $1\le k\le 3$.
	By \cref{lem:bluegreenN} applied to~$e$, we have~$k=3$ and we know the remaining edge is an inner one. By \cref{lem:partialbundleN}, we may assume $s$ crosses this edge in $D_u$. Now \cref{lem:redorangeN} implies a contradiction.
\end{lemProof}

\begin{restatable}{proper}{properTwo}\label{lem:newcompleteN}
	Let $e$ and $f$ be two edges of $T_1$ sharing an endpoint $u$.
	If $s$ has one of its endpoints $v$ in $D_u\setminus \{u\}$, it intersects all edges of $B_e\cup B_f$ inside $D_u$. Moreover, $v\in\{a_f,y_f,z_f\}$.
\end{restatable}
\begin{lemProof} 
	\cref{lem:completebundleN} implies that if $s$ intersects one edge of $B_e$ or $B_f$ inside $D_u$ it intersects all edges of this bundle inside $D_u$. Clearly, $s$ intersects the bundle that belongs to $v$; note that $v$ belongs either to $B_e$ or to $B_f$, not to both since $v\neq u$.
	For a contradiction, assume $s$ intersects only one of the bundles $B_e$ or $B_f$. We disprove this by a short case analysis:
	
	First we consider the case that $v\in\{c_e,x_e\}$. In fact, we show that this case does not occur.
	As before, let $R'$ be the region defined by the parts of $x_ey_e$ and $x_ez_e$ within $D_u$ and parts of the boundary of $D_u$; see \cref{fig:vicinRO}~(left).
	Since $s$ intersects all edges of the bundle $B_e$ within $D_u$, $s$ intersects $R'$. If $s$ starts at $c_e$, i.e., strictly inside $R'$,  it intersects only one of the boundary edges incident to $x_e$; a contradiction. If $e$ starts at $x_e$ and intersects $a_ec_e$, then $s$ is trapped inside $R'$. Thus $s$ does not intersect $a_eb_e$.
	
	Second, we consider the case that $v\in\{a_f,y_f,z_f\}$.
	In this case, $s$ intersects all edges of $B_f$ by \cref{lem:completebundleN}. Let the parts of $e$ and $f$ inside $D_u$ partition $D_u$ in a top half (outside $T_u$) and a bottom half (inside $T_u$).
	If no edge of $B_e$ is intersected, since by \cref{proper:One}, $s$ cannot leave $D_u$ inside the bundle $B_f$, $s$ can cross $f$ only by crossing all edges of $K_f$ in the top part next to the intersection with $f$ (take the union of the regions in \cref{fig:vicinRO} as forbidden region to see this); a contradiction, since $v\in\{a_f,y_f,z_f\}$ lies in the bottom half of $D_u$. Thus, $s$ also intersects all edges of $B_e$ in $D_u$.
\end{lemProof}

Now, we come to our final property.
\begin{restatable}{proper}{properThree}\label{lem:replacement}
	If there exists a new edge $s$ with vertices in $T_2$ such that $T_2\cup s$ is a thrackle, then there exists an edge $s'$ such that  $T_2\cup s'$ is a thrackle, the vertices of $s'$ belong to $T_1$, and the vertices of $s'$ do not share an edge in $T_1$.	
\end{restatable}

\begin{lemProof}
	Let $(U,V):=s$. If both $U,V$ are vertices of $T_1$, then the claim is proved. Therefore, we may assume that $U$ does not belong to $T_1$. Let $u$ denote the vertex of $T_1$ such that $U$ is contained in $D_u$; likewise, let $v$ denote the vertex of $T_1$ such that $V$ is contained in $D_v$. 
	
	\textbf{Case 1:} $U$ and $V$ lie in the same vicinity, i.e. $u=v$:
	
	By \cref{lem:newcompleteN}, it follows that $V=u$; otherwise $U,V\in\{a_f,y_f,z_f\}$ for some edge $f$ in $T_1$ incident to $u$
	which yields a contradiction to the maximality of the  Kyn\v{c}l example proved in \cref{proposition:Kyncl}. If $V=u$, then $s$ cannot leave $D_u$ in the top half, since it would have to cross all edges of $K_f$ to do so. Therefore, both ends of $s$ leave $D_u$ inside the triangular region $T_u$ of $u$ (see \cref{fig:starBelt}). $T_u$ is bounded by $e,f$ and another edge, say $l$. The only way to leave $T_u$ would be for both ends to cross $l$. This implies a contradiction.
	
	\textbf{Case 2:} $u\neq v$ share an edge in $T_1$:
	
	Remember $U\neq u$. If $V\neq v$, then by \cref{lem:newcompleteN}, $s$ intersects all edges of $B_e$ in both $D_u$ and $D_v$; a contradiction. Similarly, if $V=v$,
	then $s$ intersects all edges of $B_e$ in $D_u$ 
	and $e=uv=uV$ in $D_v$; a contradiction.
	
	\textbf{Case 3: }There is no edge incident to $u$ and $v$:
	
	We use the fact that $s$ intersects all edges present in $D_u$ (by \cref{lem:newcompleteN}) to reroute $s$ inside $D_u$. As before, let the sections of $e$ and $f$ inside $D_u$ partition $D_u$ in its top and bottom half.
	
	Let $w_1,w_2,\dots,w_{k}$ denote the sequence of intersections of $s$ with $\partial D_u$, where $w_1$ is closest to $V$ on $s$. Since the vertex $U$  of $s$ is inside $D_u$, $k$ is an odd integer. Moreover, no section $w_{2i-1}w_{2i}$ (inside $D_u$)  connects the top and bottom half: Suppose it does. Then $w_{2i-1}w_{2i}$ intersects all edges of $B_e$ or $B_f$. By \cref{lem:newcompleteN}, we know that $U\in\{a_f,y_f,z_f\}$ and hence $w_{2i-1}w_{2i}$ intersects all edges of $B_e$. However, then no section of $s$ can enter $D_u$ outside $R_f$, intersect $f$ and end at $U\in\{a_f,y_f,z_f\}$, similarly to the second case in \cref{lem:newcompleteN}.
	
	Consequently, $w_1w_2,\dots,w_{k-2}w_{k-1}$ form  pairs contained in the top or bottom part that are additionally nested since $s$ has no self-intersections. We replace the sections $w_{2i-1}w_{2i}$ of $s$ by curves close to the boundary of $D_U$ such that no edge of $D_u$ is intersected. 
	
	The last part $w_kU$ we reroute as follows, see also \cref{fig:replacement}: If $w_k$ is contained in the top half of $D_u$, we replace the part of $s$ inside $D_u$ by a straight line segment that connects $u$ and $\partial D_u\cap s$; note that this segment intersects all edges in $D_u$. If $w_k$ is contained in the bottom half of $D_u$, we replace $w_kU$ inside $D_u$ with a curve from $u$ to $\partial D_u\cap s$ as illustrated; note that this curve intersects all edges of $D_u$.

	\begin{figure}[htb]
		\centering
		\includegraphics[scale=1,page=14]{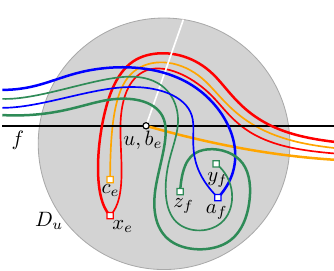}
		\includegraphics[scale=1,page=15]{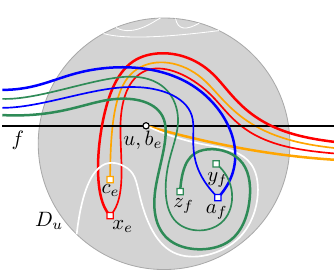}
		\caption{Illustration of \cref{lem:replacement}.}
		\label{fig:replacement}
	\end{figure}
	
	After this replacement, the new edge $s'$ intersects the same set of edges as $s$. Therefore, $T_2+s'$ is a thrackle. Moreover, the vertex $U$ of $s$ is replaced by the vertex $u$ of $s'$ where $u$ is in $T_1$. If $V\neq v$, we apply the same rerouting for the other vertex $V$ of $s$.
\end{lemProof}
\cref{lem:replacement} implies that if $T_1$ is maximal, then $T_2$ is maximal. Therefore, \cref{proposition:max1} implies that 
	$T_2$ is a maximal thrackle with $\varepsilon(T_2)=\frac{5}{6}$, which concludes the proof of Theorem~\ref{thm:4/5}.
	\end{proof}
 
 \section{Ongoing work and open problems}\label{sec:outlook}
We believe that by \emph{repeating the Kyn\v{c}l belt construction}, one obtains a class of maximal trackles such that for every $c$, there exists maximal thrackle $T$ with $\varepsilon(T)<\frac{4}{5}+c$. The idea is as follows: Since the original edges of $T_1$ are preserved in $T_2$, we can apply the Kyn\v{c}l belt construction to $T_2$ by using only the edges of~$T_1$. This results in a thrackle $T_3$. To do this, we find new, smaller vicinities around every vertex of $T_1$ which are free of other vertices and non-incident edges.
For an illustration, consider \cref{fig:doubleKynclbeltPaper}.  

\begin{figure}[htb]
	\centering
	\includegraphics[page=8]{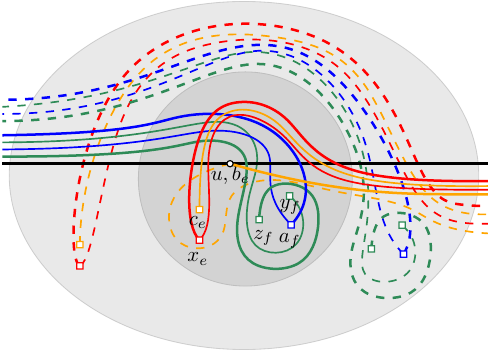}
	\caption{Applying the Kyn\v{c}l belt construction multiple times.}
	\label{fig:doubleKynclbeltPaper}
\end{figure}
By repeating the procedure $k$ times, we obtain a trackle $T_k$ with \[\varepsilon(T_k)=\frac{2n+1+4k}{2n+1+5k}=\frac{4}{5} + \frac{2n+1}{10n+5+25k} <\frac 45 + c \Leftrightarrow k>\frac {(1 -5c)(2n+1)}{25c}.\]
Showing that $T_k$ is (potentially) maximal is more involved and ongoing work, in which we are done with proving most appearing cases.

\noindent We conclude with a list of interesting open problems:~
\begin{compactitem}
	\item What is the minimal number of edges that a maximal thrackle without isolated vertices can have? Can such a maximal thrackle $T$ have $\varepsilon(T)<\frac{4}{5}$? 
	\item Is it true that for every maximal thrackle $T$ it holds that  $\varepsilon(T)>\frac{1}{2}$ or do maximal matching thrackles (other than~$K_{1,1}$) exist? 
		For geometric thrackles, this question has been very recently answered in the negative~\cite{nomatchingthrackle}, but it remains open for topological thrackles.
	\item Does Conway's conjecture hold?
\end{compactitem}

\subsection*{Acknowledgments}
O. Aichholzer and B. Vogtenhuber partially supported by Austrian Science Fund within the collaborative DACH project \emph{Arrangements and Drawings} as FWF project \mbox{I 3340-N35}. Travel costs of Felix Schr\"oder were supported by DFG Grant FE 340/12-1.

This research was initiated during the 15th European Research Week on Geometric Graphs (GGWeek 2018) at  Haus Tornow am See (M\"arkische Schweiz, Germany) and Freie Universit\"at Berlin.
The workshop was supported by the Deutsche Forschungsgemeinschaft (DFG) through the Research Training Network \emph{Facets of Complexity} and the collaborative DACH project
\emph{Arrangements and Drawings}.
We thank the organizers and all participants for the stimulating atmosphere.
In particular, we thank Andr\'e Schulz for proposing the study of maximal thrackles as a research question,
and Viola M\'esz\'aros and Stefan Felsner for joining some of our discussions and contributing valuable ideas.

\bibliographystyle{splncs04}
\bibliography{bib}

\end{document}